\documentclass[11pt]{amsart}
\usepackage{amssymb,amsthm,amsmath,amscd}
\usepackage{epsfig}
\usepackage{subfigure}
\usepackage{fancyhdr}
\usepackage[top=1in,bottom=1.5in,left=1.5in,right=1.5 in]{geometry}
\usepackage[all]{xy}

\usepackage{enumerate}

\newcommand{\be}{\begin{eqnarray}}
\newcommand{\ee}{\end{eqnarray}}
\newcommand{\ben}{\begin{eqnarray*}}
\newcommand{\een}{\end{eqnarray*}}

\numberwithin{equation}{section}

 \DeclareMathOperator{\diam}{diam}

\newcommand{\R}{{\mathbf R}}

\def\R{{\mathbb{R}}}
\def\HH{{\mathcal{H}_d}}

\def\R{{\mathbb{R}}}

\def\d{{m}}

\def\E{{\mathcal{E}}}

\newtheorem{theorem}{Theorem}

\newtheorem{corollary}[theorem]{Corollary}

\newtheorem{lemma}[theorem]{Lemma}

\newtheorem*{theorem*}{Theorem}
\newtheorem*{remark*}{Remark}

 \title[Quasi-uniformity of Minimal Weighted Energy Points ]{Quasi-uniformity of Minimal Weighted Energy Points on Compact Metric Spaces}

\author{D. P. Hardin,   E. B. Saff, and J. T. Whitehouse} 
\thanks{\noindent The research of all authors was supported, in part, by the U. S. National Science Foundation under grants DMS-0808093 and DMS-1109266. }

\date{\today}

\address{D. P. Hardin, E. B. Saff, and J. T. Whitehouse:
Center for Constructive Approximation,
Department of Mathematics,
Vanderbilt University,
Nashville, TN 37240,
USA }
\email{Doug.Hardin@Vanderbilt.Edu}
\email{Edward.B.Saff@Vanderbilt.Edu}
\email{Tyler.Whitehouse@gmail.com}


\begin{document}
 \begin{abstract}
 For a closed subset $K$ of a compact metric space $A$ possessing
 an $\alpha$-regular measure $\mu$ with $\mu(K)>0$, we prove that  whenever $s>\alpha$,
 any sequence of weighted minimal Riesz $s$-energy configurations
 $\omega_N=\{x_{i,N}^{(s)}\}_{i=1}^N$ on $K$ (for `nice' weights) is quasi-uniform  in the sense that
 the ratios of its mesh norm to separation distance remain bounded as $N$ grows large. Furthermore, if $K$ is  an $\alpha$-rectifiable compact subset of Euclidean space ($\alpha$ an integer)  with positive and finite $\alpha$-dimensional Hausdorff measure, it is possible to generate  such a quasi-uniform sequence of configurations
 that also has (as $N\to \infty$) a prescribed positive continuous
 limit distribution with respect to $\alpha$-dimensional Hausdorff measure.  \end{abstract}
 \maketitle

\section{Introduction}

Let $A$ be a compact infinite metric space with metric $\d: A\times A\rightarrow[0,\infty)$ and  let $\omega_N=\{x_i\}_{i=1}^N\subset  A$ denote a configuration of $N\ge 2$ points in $A$.
We are chiefly concerned with  two `quality' measures of $\omega_N$; namely, the {\em separation distance of $\omega_N$}   defined by
\begin{equation}\delta(\omega_N):=\min_{1\leq i\not= j\leq N}\d(x_i,x_j),
\end{equation}
and    the   {\em mesh norm}  {\em of $\omega_N$ with respect to   $A$}   defined by
 \begin{equation}\rho(\omega_N,A):=\max_{y\in A}\min_{1\leq i\leq N}\d(y,x_i).
 \end{equation}
 This quantity is also known as the  {\em fill radius} or {\em covering radius} of $\omega_N$ relative to $A$.
The optimal values of these quantities are also of interest and we consider, for $N\ge 2$,  the {\em $N$-point best-packing
 distance on $A$} given by $$\delta_N(A):=\max \{\delta(\omega_N)\colon \omega_N\subset  A,\, |\omega_N|=N\},$$
 and the {\em $N$-point mesh norm}  of  $A$  given by
 $$\rho_N(A):=\min  \{\rho(\omega_N,A)\colon \omega_N\subset  A,\, |\omega_N|=N\},$$
 where $|S|$ denotes the cardinality of set $S$.

  In the theory of approximation and interpolation (for example, by splines or radial basis functions (RBFs)), the separation distance is often associated with some measure of `stability' of the approximation, while the mesh norm arises in the  error of the approximation. In this context, the  {\em mesh-separation ratio} (or {\em mesh ratio}) $$\gamma(\omega_N,A):=\rho(\omega_N,A)/\delta(\omega_N),$$
 can be regarded as a `condition number' for $\omega_N$ relative to $A$.   If $\{\omega_N\}_{N=2}^\infty$ is a sequence of $N$-point configurations such that $\gamma(\omega_N,A)$ is uniformly bounded in $N$, then the sequence is said to be {\em quasi-uniform on $A$}.
    Quasi-uniform sequences of configurations are important for a number of methods involving RBF approximation and interpolation (see \cite{ FW, LGSW, P, S}).

 We remark that in some cases it is easy to obtain positive lower bounds for the mesh-separation ratio.  For example, if $A$ is connected, then $\gamma(\omega_N,A)\ge 1/2$.   Furthermore, letting
 $$B(x,r)=\{y\in A: m(y,x)\leq r\}$$ be the closed ball in $A$ with center $x$ and radius $r$,
  then
 $\gamma(\omega_N,A)\ge \beta/2$  for any $N$-point configuration $\omega_N\subset  A$ whenever $A$ and $\beta\in(0,1)$ have the property that for any $r\in (0,\diam (A)]$ and any $x\in A$, the annulus $B(x,r)\setminus B(x ,\beta r)$
 is nonempty.     The diameter of $A$ is defined by $$\diam(A):=\max\{m(x,y)\colon x\in A,\ y\in A\}.$$

  In this paper we consider the separation distance and mesh norm of finite point configurations in $A$
  that   minimize certain weighted energy  functionals.  We call  $w:A\times A\rightarrow[0,\infty)$ an  {\em SLP weight on $A$} if it is  symmetric and lower semi-continuous on $A\times A$ and is positive on the diagonal, $D(A)$, of $A\times A$.
For  $s>0$ and  a collection of $N\ge 2$ distinct points $\omega_N=\{x_1,\ldots,x_N\}\subset  A$,  the {\em $(s,w)$-energy  of $\omega_N$} (also known as the {\em weighted Riesz $s$-energy}) is
\begin{equation}\label{weighted}E_{s}^{w}(\omega_N):=\sum_{i\not=j}\frac{w(x_i,x_j)}{\d(x_i,x_j)^s}=
\sum_{i=1}^N\sum_{\substack{j=1\\j\neq i}}^N\frac{w(x_i,x_j)}{\d(x_i,x_j)^s},
\end{equation}
and  we denote the {\em minimal $N$-point $(s,w)$-energy of  $A$} by
\begin{equation}\label{mins}\E_{s}^{w}(N,A):= \inf\{E_{s}^{w}(\omega_N):\omega_N\subset  A, \, |\omega_N|=N\}.
\end{equation}
Since $A$ is compact and the energy $E_{s}^{w}(\omega_N)$ is lower semi-continuous, there exists at least one $N$-point configuration $\omega_N^*\subset  A$ such that $E_{s}^{w}(\omega_N^*)=\E_{s}^{w}(N,A)$.  We refer to such an $\omega_N^*$ as an {\em $N$-point  $(s,w)$-energy minimizing configuration on $A$}.  The asymptotics  as $N\to \infty$
of  $N$-point $(s,w)$-energy minimizing configurations
and their energies  are investigated in \cite{transactions,advances} for $d$-rectifiable sets $A\subset  \R^p$ and $s>d$ (see further discussion in the next section).

In our results we shall require that  $A$  is either $\alpha$-regular or upper $\alpha$-regular as we next describe.   For a positive Borel measure $\mu$ supported on $A$ and $\alpha>0$, we say that $\mu$ is \emph{upper $\alpha$-regular} if there is some finite  constant $C_0$ such that%
\begin{equation}\label{upregularity} \mu(B(x,r))\leq C_0 \,r^\alpha \qquad (x\in A,\, 0<r\leq\diam(A)),\end{equation}
  and we say that $\mu$ is \emph{lower $\alpha$-regular} if there is some  positive constant $c_0$ such that
  \begin{equation}\label{lowregularity}
  c_0^{-1} \,r^\alpha\le \mu(B(x,r))  \qquad (x\in A,\, 0<r\leq\diam(A)).
  \end{equation}
   We shall refer to $A$ as an   \emph{upper $\alpha$-regular metric space}  if there exists
an upper   $\alpha$-regular measure $\bar{\mu}$ on $A$ such that $\bar{\mu}(A)>0$ and shall refer to $A$ as a \emph{lower $\alpha$-regular metric space}  if there exists
a lower  $\alpha$-regular measure $\underline{\mu}$ on $A$ such that $\underline{\mu}(A)<\infty$. (Obviously, if $A$ is upper $\alpha$-regular then $A$ has infinitely many points.)
If $A$ supports a measure that is both upper and lower $\alpha$-regular, then we say that $A$ is an \emph{$\alpha$-regular metric space}.
If $A$ is $\alpha$-regular, then it is not difficult to show that the Hausdorff dimension of
$A$,  $\dim_\mathcal{H}A$, equals $\alpha$ (cf. \cite{Heinonen,M}).  Furthermore,  the $\alpha$-dimensional Hausdorff measure of $A$, $\mathcal{H}_\alpha(A)$,
is positive and finite.

Many of the constants appearing in this paper,   either explicitly or implicitly involve the upper and lower regularity constants $C_0$ and $c_0$ appearing in  \eqref{upregularity} and \eqref{lowregularity}.  However, in certain cases we are interested in `local' regularity estimates (i.e., for $r$ small) which can substantially improve our explicit estimates for particular metric spaces of interest (e.g., $A$ is the sphere $S^d$ with the Euclidean metric).  Specifically, if $\bar{\mu}$ is an upper $\alpha$-regular measure, $\underline{\mu}$ is a lower $\alpha$-regular measure and $r^*>0$, we define
\begin{align}\label{localreg}
\begin{split}
C_0(r^*)&:=\sup \{\bar{\mu}(B(x,r))/r^\alpha\colon x\in A, \, 0<r\le r^*\},\\
c_0(r^*)^{-1}&:=\inf \{\underline{\mu}(B(x,r))/r^\alpha\colon x\in A, \, 0<r\le r^*\}.
\end{split}
\end{align}We note that both $C_0(r^*)$ and $c_0(r^*)$ are increasing in $r^*$, and we make the definitions
\begin{align}\label{localreg-0}
\begin{split}
C_0(0)&:=\lim_{r^*\to 0^+}C_0(r^*),\\
c_0(0)&:=\lim_{r^*\to 0^+}c_0(r^*).
\end{split}
\end{align}Furthermore, if $A$ is a compact (i.e., without boundary), $C^1$, $d$-dimensional manifold and $\mu=\mathcal{H}_d$, then $C_0(0)\cdot c_0(0)=1$.
For the largest length scale of interest, with a slight abuse of notation,   the global  constants for $\bar{\mu}$ and $\underline{\mu}$, respectively, are related by $C_0=C_0(\diam(A))$ and $c_0=c_0(\diam(A))$.

One may obtain simple upper bounds  for $\delta_N(A)$ (respectively, lower bounds for $\rho_N(A)$)  in the case that  $A$ is   lower (respectively, upper) $\alpha$-regular.  Specifically, if $A$ is lower $\alpha$-regular then  there is a constant $c_A<\infty$ such that
 \begin{equation}\label{seplowbnd}
 \delta_N(A)\le c_AN^{-1/\alpha},\qquad (N\ge 2),
 \end{equation}
while if $A$ is upper $\alpha$-regular then  there is a constant $\tilde c_A>0$ such that
\begin{equation}\label{meshupbnd}\rho_N(A)\ge \tilde c_AN^{-1/\alpha},\qquad (N\ge 2).
 \end{equation}
  The bound \eqref{seplowbnd} is a consequence of the facts   that the balls $\{B(x,\delta (\omega_N)/2)\colon x\in\omega_N\}$   are pairwise disjoint and  that there exists a  lower $\alpha$-regular measure $\underline{\mu}$ with   $\underline{\mu}(A)<\infty.$  Similarly, if $A$ is upper $\alpha$-regular, then the bound \eqref{meshupbnd} follows from the covering property of  the balls $\{B(x,\rho (\omega_N,A))\colon x\in\omega_N\}$ and the existence of an upper $\alpha$-regular measure $\bar{\mu}$ with   $\bar{\mu}(A)>0.$

 The main result of this paper, given in Theorem~\ref{mainResult}, is that a sequence of $N$-point $(s,w)$-energy minimizing configurations on an $\alpha$-regular compact metric space $A$ is quasi-uniform on $A$  whenever $s>\alpha$.  As an application, we deduce that, if $A\subset  \R^p$ is $d$-rectifiable for some integer $0< d\le p$
with $\mathcal{H}_d(A)>0$, then a quasi-uniform sequence of $N$-point configurations on $A$ can be found that has a prescribed  bounded positive density on $A$ (see Corollary~\ref{wrho} and the discussion preceding it).

\section{Main Results \label{mainResults}}

We first consider the separation distance of $(s,w)$-energy minimizing configurations on an upper $\alpha$-regular compact metric space $A$.  For these separation results, we consider symmetric weight functions $w$ such that  $\|w(\cdot,x)\|_{L_{p}(\mu)}$ is uniformly bounded on $A$ for some $1<p\leq\infty$.  Here we use the standard notation,
$$\|f\|_{L_{p}(\mu)}:=\begin{cases}\left(\int_A |f|^{p}\,d \mu\right)^{1/{p}},  & 1\leq p<\infty  ,\\    \text{$\mu$-ess  sup $|f|$}, &   p=\infty,
\end{cases}
$$
where $\mu$ is a positive Borel measure and $f$ is a Borel measurable function on $A$.

The following theorem extends a result~\cite[Theorem~4]{transactions} to a more general class of weight functions and to more general compact metric spaces.
\begin{theorem}\label{separation} Let  $A$ be a    compact,  upper $\alpha$-regular metric space with respect to   $\bar \mu$  and let $w$ be an SLP weight on $A$ such that
$\|w(\cdot,x)\|_{L_{p_0}(\bar\mu)}$ is uniformly bounded on $A$ for some $1<p_0\leq\infty$.  Suppose
$1<p\leq p_0$,  $s>\alpha(1-1/p)$, and  $N\ge 2$.  If  $\omega_N^*$ is an $N$-point $(s,w)$-energy minimizing configuration on $A$,
then
\begin{equation}\label{prop-sep}\delta(\omega_N^*)\geq C_1\, N^{-\left(\frac{1}{\alpha} +\frac{1}{sp}\right)} \qquad (N\ge 2),
\end{equation}where $C_1$ is a constant independent of $N$ indicated below in~\eqref{C1def}.
 \end{theorem}

Taking $w$ bounded and setting $p=\infty$ in Theorem~\ref{separation}   produces the following result.
 \begin{corollary}\label{seps}Suppose  $A$ is a  compact,   upper $\alpha$-regular   metric space and  $w $ is a bounded SLP weight on $A$, and let  $s>\alpha$.  If $\omega^*_N$ is an $N$-point $(s,w)$-energy minimizing  configuration on $A$, then
\begin{equation}\label{seps1}
 \delta(\omega^*_N)\ge C_2N^{-1/\alpha} \qquad (N\ge 2),
 \end{equation}where $C_2$ is a constant independent of $N$.
 Consequently,
 \begin{equation}\label{sep-const}
 \delta_N(A)\geq C_2 N^{-1/\alpha}\qquad (N\ge 2).
 \end{equation}

 For the unweighted case $w\equiv 1$, the constant $C_2$ satisfies
 \begin{equation}
 \label{C2bnd} C_2\ge \left[\frac{\bar{\mu}(A)}{C_0}\left (1-\frac{\alpha}{s}\right)\right]^{1/\alpha}\left(\frac{\alpha}{s}\right)^{1/s},
 \end{equation}
where $C_0=C_0(\diam(A))$.
 \end{corollary}
  We note that if $A$ in Corollary~\ref{seps} is $\alpha$-regular, then by inequality~\eqref{seplowbnd} we see that   $N$-point $(s,w)$-energy minimizing  configurations on $A$
 have the best possible order of separation as $N\to \infty$.

With respect to the separation constant of~\eqref{C2bnd}, if $d\geq 2$ and  $A=\mathbb{S}^d$ with $
 \sigma_d$ denoting the uniform probability distribution on $\mathbb{S}^d,$ then we can get an
 explicit lower bound for $C_2$ by calculating the regularity constant $C_0.$ As stated in ~\cite{kuij}, for $x\in\mathbb{S}^d$,  $0\leq r\leq 2$, and
\begin{equation}\label{gamma-d}\gamma_d :=\frac{\Gamma\left(\frac{d+1}{2}\right)}{\Gamma(d/2)\Gamma(1/2)},
\end{equation}
there holds
 $$\sigma_d(r):=\sigma_d(B(x,r))=\gamma_d\int^1_{1-r^2/2}(1-t^2)^{d/2-1}dt$$ from which it follows that
 $$\sigma_d(r)\leq \frac{\gamma_d}{d}r^d,$$
and, as  $r\to0^+,$
 $$ \sigma_d(r)=\frac{\gamma_d}{d}r^d+\mathcal{O}(r^{d+2}).$$
Therefore, for the uniform probability distribution on $\mathbb{S}^d$, the global upper regularity constant is
 \begin{equation}\label{C-0-uniform}C_0=\sup_{0<r\leq 2}\frac{\sigma_d(r)}{r^d}=\frac{\gamma_d}{d},\end{equation}
%
%
and when applied to~\eqref{C2bnd} we obtain
\begin{equation}\label{C2bnd-unif}C_2\geq \left(\frac{d}{\gamma_d}\right)^{1/d}\left(1-\frac{d}{s}\right)^{1/d}\left(\frac{d}{s}\right)^{1/s}.
\end{equation}
With this lower bound for $C_2$,~\eqref{seps1} becomes
\begin{equation}\label{sep-uniform}
 \delta(\omega^*_N)\ge \left(\frac{d}{\gamma_d}\right)^{1/d}\left(1-\frac{d}{s}\right)^{1/d}\left(\frac{d}{s}\right)^{1/s}N^{-1/d} \qquad (N\ge 2,\,s>d),
\end{equation} and, on letting $s \to \infty,$ we deduce for the $N$-point best-packing distance
 $$\delta_N(\mathbb{S}^d)\geq  \left(\frac{d}{\gamma_d}\right)^{1/d}N^{-1/d} \qquad (N\ge 2,\,s>d).
 $$
 A less explicit lower bound for the separation constant of minimal energy points for $s>d$ on $\mathbb{S}^d$ was obtained in~\cite[Corollary 4]{kuij}.

%


We next consider the mesh norm of $(s,w)$-energy minimizing configurations on an $\alpha$-regular compact metric space $A$. In this case we require that the weight function $w$ be bounded.

\begin{theorem}\label{let} Let $A$ be a compact, $\alpha$-regular  metric space
with   respect to the measure $\mu$ and   $K\subset  A$  be  a compact set of positive $\mu$-measure. Let  $w$ be  a bounded SLP weight on $K$.  If  $s>\alpha$
and  $\omega_N^*$ is an $N$-point $(s,w)$-energy minimizing configuration on $K$,
then
 \begin{equation}\label{mesh-1}\rho(\omega_N^*,K)\leq C_3\, N^{-1/\alpha}\qquad (N\ge 2),
\end{equation}where $C_3$ is a constant independent of $N$ given below in~\eqref{C3def}.
\end{theorem}

Theorem~\ref{let} substantially extends a result of \cite{maym} that holds for unweighted energy minimizing point configurations when $K\subset  \R^p$ is restricted to be
the finite union of bi-Lipschitz images of compact sets in $\R^d$.

We remark that for $K$ and $A$ as in Theorem~\ref{let}, the set $K$ need not   inherit the lower $\alpha$-regularity of $A$.
However, since ${\mu}(K)>0$, we do have that $K$   is an upper $\alpha$-regular metric space and, consequently,  there is a constant $\tilde{c}_K>0$ such that \eqref{meshupbnd}
holds with $A$ replaced by $K$.
 Hence,  the inequality~\eqref{mesh-1} has the best possible order  with respect to $N$.

Taking $w\equiv 1$ in Theorem~\ref{let} immediately yields the following.
\begin{corollary}\label{lucky} Let $A$ be a compact, $\alpha$-regular  metric space
with   respect to the measure $ \mu$ and let $K\subset  A$  be a compact set of positive $\mu$-measure.  Then there exists a constant $C_4$ such that
$$\rho_N(K)\leq C_4\, N^{-1/\alpha}\qquad (N\ge 2).$$

\end{corollary}

\medskip

Combining   Corollary~\ref{seps}  and  Theorem~\ref{let} we obtain our main result.

\begin{theorem}\label{mainResult} Let $A$ be a compact,   $\alpha$-regular  metric space
with   respect to the  measure $\mu$ and let $K\subset  A$  be a compact set of positive $\mu$-measure.  Furthermore, let  $w$ be a bounded SLP weight on $K$,  and for $s>\alpha$ and $N\geq2$,
    let  $\omega_N^*$ be an $N$-point $(s,w)$-energy minimizing configuration on $K$.
Then $\{\omega_N^*\}_{N=2}^\infty$ is quasi-uniform on $K$.
\end{theorem}

We remark that there are  $\alpha$-regular sets  $A$  and values of $s<\alpha$   for which (unweighted) $(s,1)$-energy minimizing configurations
on  $A$  have a mesh-separation ratio that
goes to $\infty$ with $N$.  One such example given in \cite{BrauHS} is a `washer' $A$ obtained by revolving a certain rectangle about an axis parallel to one of its sides,
where it turns out that for $s<1/3$, the support of the limit distribution of the $(s,1)$-energy minimizing configurations on $A$ omits an open subset of $A$.  Also, for
the logarithmic energy which corresponds to  $s=0$, it is shown in \cite{HSS} that, for $w\equiv 1$, the support of the limit distribution of  the log-energy minimizing configurations on a torus in $\R^3$ is only supported on the  positive curvature portion of the torus, so that the mesh-separation ratio for such configurations is again unbounded as $N\to \infty$.
Examples also abound in one dimension.  For the logarithmic energy, it is well-known \cite[Sections 6.7 and 6.21]{Sze}
that for $A=[-1,1]$ and $w\equiv 1$ the minimum energy points are zeros of Jacobi orthogonal polynomials (together with
$\pm 1$) that have separation distance of precise order $1/N^2$ and mesh norm of precise order $1/N$, so that the mesh-separation ratio grows like $N$.

One of our main motivations for considering weighted minimum energy configurations is that for a large class of sets $A$ one can design a weight function $w$ so that a sequence of $N$-point $(s,w)$-energy  minimizing configurations have a specified
limiting  density on $A$ as $N\to \infty$.
The following result is a consequence of Theorem~\ref{mainResult} and~\cite[Corollary 2]{transactions}.
Recall that a set in $\R^p$ is {\em $d$-rectifiable} if it is the Lipschitz image of a bounded set in $\R^d$.

\begin{corollary}\label{wrho}
Let $d\leq p$ and $A\subset  \R^{p}$ be a compact, infinite set that is   $d$-rectifiable and lower $d$-regular with respect to $\mathcal{H}_d$  for some integer $d$.
Suppose $\sigma$ is a probability density on $A$ that is continuous almost everywhere with respect to $\HH$  and is bounded above and below by positive constants.  Let $s>d$ and $w:A\times A\to [0,\infty)$ be
  given by
  \begin{equation} \label{wrhodef}
w(x,y):=(\sigma (x)\sigma (y))^{-s/2d}.\end{equation}  For   $N\ge 2$, let $\omega_N^*$ be an $N$-point $(s,w)$-energy minimizing configuration on $A$.  Then $\{\omega_N^*\}_{N=2}^\infty$ is quasi-uniform on $A$ and the sequence of normalized counting measures associated with   the    $\omega_N^*$'s converges
  weak-star
   (as $N\to \infty$) to $\sigma\,  \mathrm{d}\mathcal{H}_d$.
\end{corollary}

For $A$   an   infinite, compact, metric space
  and $s>0$, let $\omega_N^s$ be an $N$-point $(s,1)$-energy minimizing configuration on $A$.  Furthermore, let $\nu_N$ be a cluster point (in the product topology on $A^N$) of  $\omega_N^s$ as $s\to \infty$.  As we now show,  $\nu_N$ must be an {\em $N$-point best-packing configuration on $A$}, that is, $\delta(\nu_N)=\delta_N(A)$.  For this purpose,  let $\tilde\omega_N$ be an   $N$-point best-packing configuration on $A$.  Then we have
$$
\delta(\omega_N^s)^{-s}\le \mathcal{E}_s^1(N,A)\le E_s^1(\tilde{\omega}_N)\le N(N-1)\delta_N(A)^{-s},
$$
and so
$$
(N(N-1))^{-1/s}\delta_N(A)\le \delta(\omega_N^s) \le \delta_N(A),
$$
which gives
\begin{equation}\label{nuN}
\lim_{s\to \infty} \delta(\omega_N^s) =\delta_N(A).
\end{equation}
Since $\omega_N^{s_j}\to \nu_N$ for some subsequence $s_j\to \infty$, it follows from \eqref{nuN} and  continuity that $\delta(\nu_N)=\delta_N(A)$ and so $\nu_N$ is an  $N$-point best-packing configuration on $A$.

 In general, it is not true that a sequence of $N$-point best-packing configurations in $A$ is quasi-uniform on $A$ (e.g., if $A$ is the classical $(1/3)$-Cantor set in [0,1]  together with any point outside this interval).  However,  for $A$ as in Theorem~\ref{mainResult}, it turns out that by using $(s,1)$-energy minimizing configurations on $A$ and taking $s\to \infty$ we can construct a  sequence of $N$-point best-packing configurations in $A$ that is also quasi-uniform on
$A$.

\begin{theorem}\label{QUBP} Let $A$ be a compact,   $\alpha$-regular  metric space
with  respect to the measure $\mu$ and let $K\subset  A$  be a compact set of positive $\mu$-measure.
For $N\ge 2$,   let  $\nu_N$ be a cluster point of a family of $N$-point $(s,1)$-energy minimizing configurations on $K$ as $s\to \infty$.
Then $\{\nu_N\}_{N=2}^\infty$ is a sequence of $N$-point best-packing configurations on $K$ that   is also quasi-uniform on $K$.

Furthermore,   the mesh-separation ratios satisfy
\begin{equation}\label{bigger-dan}
\limsup_{N\to\infty}\gamma(\nu_N,K)\le  \ 2\left(\frac{\mu(A)}{\mu(K)}\right)^{1/\alpha}[c_0(0)\, C_0(0)]^{1/\alpha},
\end{equation}
where $c_0(0)$ and $C_0(0)$ are given in \eqref{localreg-0} for the set $A$.\footnote{{\em Added in proof:} In the manuscript [1], the first two authors together with A. Bondarenko have recently proved under more general conditions that the right-hand side of \eqref{bigger-dan} can be replaced by 1.}
\end{theorem}We note that the constant on the right-hand side of~\eqref{bigger-dan} is at least $2$
per~\eqref{localreg} and~\eqref{localreg-0}.
One can also establish an analogous result concerning the existence of quasi-uniform sequences of {\em weighted} best-packing configurations
(cf. \cite{BSH3}).  We leave this extension to the reader.

In comparison with \eqref{bigger-dan}, we remark that   one can construct examples of metric spaces $A$ having  $n$-point best-packing configurations with arbitrarily large mesh-separation ratio.

We conclude this section with further references to related results.
   Separation theorems for the case $s\le d=\dim_\mathcal{H}(A)$ have been established only for rather special sets and values of $s$.
Dahlberg  \cite{Da} proved that (unweighted) optimal $((p-2),1)$-energy configurations $\omega_N^*$ on $A$  are {\em well-separated} (i.e., they satisfy
$\delta(\omega_N^*)\ge C N^{-1/d}$ for some positive constant $C$) if
$A\subset  \R^p$ ($p\ge 3$) is a smooth $d=p-1$ dimensional closed surface in $\R^p$  that separates $\R^p$ into two components.
For the critical value $s=d$ and $A$ a $d$-rectifiable subset of a smooth $d$-dimensional manifold in $\R^p$, it is shown in \cite{transactions} that
the following weaker separation result holds
\begin{equation}\label{weaksep}
\delta(\omega_N^*)\ge C (N\log N)^{-1/d},
\end{equation}
for some positive constant $C$.

For the case that  $A=\mathbb{S}^d$, the $d$-dimensional unit sphere in $\R^{d+1}$,  well-separation was proved  in \cite{KSS} for the range of values
$d-1<s<d$ and further extended by Dragnev and Saff \cite{DS} to the range
$d-2<s<d$ with explicit estimates for the separation constant $C$.  Well-separation  for $s=d-2$ and $d\ge 3$ was established in \cite{maym}.

Thus,  for the important case  of $A=\mathbb{S}^2$ it is known that optimal $s$-energy configurations on $\mathbb{S}^2$ are
well-separated for all nonnegative values of $s\neq 2$ (well-separatedness for  $s=0$ was established in \cite{RSZ}; see also \cite{D}); for the critical value $s=2$, the only known separation results are of the weak form given in \eqref{weaksep}.

Much less is known with regard to covering (mesh norm) theorems  in the case that $s\le d$ (see \cite[Sec.~1.3]{Sch}).

\section{Proofs}

In the proofs we shall need that an SLP weight $w$ is bounded below in a neighborhood of the diagonal $D(A)$.
Indeed, the positivity and  lower semi-continuity of $w$ on $D(A)$  and the compactness of $A$ imply that there are positive numbers $\eta$ and $\kappa$ such that
\begin{equation}\label{etakappa}
w(x,y)\ge \eta \qquad (x,y\in A, \,  m(x,y)\le \kappa).
\end{equation}

\begin{proof}[Proof of Theorem~\ref{separation}]

The initial part of this   argument proceeds as in \cite{kuij}.  Let $N\ge 2$ be fixed and let $\omega_N^*=\{x_1,\ldots,x_N\}\subset  A$ be a fixed $(s,w)$-energy minimizing configuration in $A$. For  $x\in A$ and $1\leq i\leq N$, let
$$U_i(x):=\sum_{\substack{  j=1\\ j\not=i}}^N\frac{w(x,x_j)}{\d(x,x_j)^s}.$$
Since $\omega_N^*$ is a minimizing configuration we have the lower bound
\begin{equation}\label{minimum}U_i(x_i)\leq U_i(x)\textup{ for all } x\in A.\end{equation}

Fix    $r_1\le\diam(A)$  such that
\begin{equation}\label{r1def}
\bar{\mu}\left(\bigcup_{j=1}^NB(x_j,r_1)\right)\ge \bar{\mu}(A).\end{equation}The radius $r_1$ can clearly be chosen independent of $N$, for example $r_1=\diam(A)$, and we note for future reference that it suffices to take $r_1>\rho(\omega_N^*,A).$  For the rest of this proof we fix $r_1=\diam(A)$.

Now let $0<\theta<1$ and
define
\begin{equation}\label{r-0}
r_0:=\left( \frac{\theta\bar{\mu}(A)}{N\, C_0(r_1)}\right)^{1/\alpha},\end{equation}
 where $C_0(r_1)=C_0$ is the  upper regularity constant of $\bar{\mu}$ as in  \eqref{localreg}. We note that $r_0< r_1$
as can be seen from the fact that $\bar{\mu}(A)\leq C_0(r_1)r_1^\alpha.$

 For $B(x,r_0,r_1):=B(x,r_1)\setminus B(x,r_0)$,   let
$$D:=  \bigcup_{j=1}^NB(x_j,r_0,r_1).$$ Using  the upper regularity of $\bar{\mu}$ and
\eqref{r1def} we see that
$$\bar{\mu}(D)\geq \bar{\mu}(A)-\sum_{j=1}^N\bar{\mu}(B(x_j,r_0))\geq (1-\theta)\bar{\mu}(A)>0,$$ and thus by inequality~\eqref{minimum} we have
\begin{equation}\label{decomp2}U_i(x_i)\leq\frac{1}{\bar{\mu}(D)}\int_{D}U_i(x)\, d\bar{\mu}(x)\leq
\frac{1}{(1-\theta)\bar{\mu}(A)}\sum_{\substack{j=1\\j\not=i}}^N\int_{ B(x_j,r_0,r_1)}\frac{w(x,x_j)}{\d(x,x_j)^s}\, d\bar{\mu}(x).
\end{equation}

Applying H\"older's inequality with $1/q=1-1/p$  we  obtain
\begin{equation}\label{finished}U_i(x_i)\leq\frac{1}{(1-\theta)\bar{\mu}(A)}\sum_{\substack{  j=1 \\ j\not=i}}^N\|w(\cdot,x_j)\|_{L_p(\bar{\mu})}\left(\int_{ B(x_j,r_0,r_1)}\frac{1}{\d(x,x_j)^{sq}}\, d\bar{\mu}(x)\right)^{1/q}.\end{equation}
Converting the integral on the right-hand side of~\eqref{finished} to the appropriate integral of the distribution function, and noting that $sq>\alpha$ by assumption,   we have
\begin{align}\label{truncate}\int_{ B(x_j,r_0,r_1)}\frac{1}{\d(x,x_j)^{sq}}\, d\bar{\mu}(x)&=\int_0^\infty\bar{\mu}\left(\{x\in B(x_j,r_0,r_1): m(x_j,x)^{-sq}>t\}\right)dt
\\&\leq\int_{r_1^{-sq}}^{r_0^{-sq}}\bar{\mu}\left(B(x_j,t^{-1/sq})\right)\, dt\nonumber\\& \leq \frac{C_0(r_1)\, sq}{sq-\alpha}\, r_0^{\alpha-sq}\nonumber\\
&=\frac{C_0(r_1)\, sq}{sq-\alpha}\,\left(\frac{\theta\bar{\mu}(A)}{N \, C_0(r_1)}\right)^{1-(sq)/\alpha},\nonumber
\end{align}
which, combined with  \eqref{finished},    gives
\begin{align}\label{hahaha}
\begin{split}
U_i(x_i)&\leq \frac{  \|w\|_{p,\infty}}{(1-\theta)\bar{\mu}(A)}\left(\frac{C_0(r_1)\, sq}{sq-\alpha}\right)^{1/q}(N-1)\, \left(\frac{\theta\bar{\mu}(A)}{N \, C_0(r_1)}\right)^{1/q-s/\alpha}\\ &<  \frac{1}{\bar{\mu}(A)}\left( \frac{C_0(r_1) }{\bar{\mu}(A)}\right)^{s/\alpha}\, \left( \frac{\|w\|_{p,\infty}}{(1-\theta)\theta^{s/\alpha-1/q}}\right)\left(\frac{sq   \bar{\mu}(A)}{sq-\alpha}\right)^{1/q}N^{1/p+s/\alpha},\end{split}
\end{align}where $\|w\|_{p,\infty}:=\sup_{x\in A}\|w(\cdot,x)\|_{L_p(\bar{\mu})}<\infty$.

Choosing
\begin{equation}\label{theta-0}
\theta_0:=\frac{sq-\alpha}{sq-\alpha+\alpha q}  =\left(   {\frac{s}{\alpha}-\frac{1}{q}}\right) \left({ \frac{s}{\alpha}+\frac{1}{p}}\right)^{-1}<1,\end{equation}
which minimizes the
right-hand side of \eqref{hahaha} with respect to $\theta$, we obtain
\begin{equation}
U_i(x_i)\le c_1N^{s/\alpha+1/p},
\end{equation}
where after a bit of arithmetic we have  \begin{equation}\label{c1def}
c_1:=\|w\|_{p,\infty}\left( \frac{C_0(r_1) }{\bar{\mu}(A)}\frac{s/\alpha+1/p}{s/\alpha-1/q}\right)^{s/\alpha}\,  \left( \frac{ s/\alpha+1/p }{ \bar{\mu}(A)}\right)^{1/p}\left( {   s/\alpha}\right)^{1/q}.
\end{equation}
Next, select  the indices $1\leq i_s\not= j_s\leq N$ so that
$\delta(\omega_N^*)=\d(x_{i_s},x_{j_s})$  and let $\kappa$ and $\eta$ be as in \eqref{etakappa}.  If $\delta(\omega_N^*)\le \kappa$, then
\begin{equation}\label{first}\frac{\eta}{\delta(\omega_N^*)^s}\leq\frac{w(x_{i_s},x_{j_s})}{\d(x_{i_s},x_{j_s})^s}\leq U_{i_s}(x_{i_s})\leq c_1 N^{s/\alpha+1/p},\end{equation}and therefore
$$\delta(\omega_N^*)\geq\left(\frac{\eta}{c_1}\right)^{1/s}\, N^{-\frac{1}{\alpha}-\frac{1}{sp}}.$$
Hence, \eqref{prop-sep} holds with \begin{equation}\label{C1def}C_1:=\min\{\kappa,\ ({\eta}/{c_1})^{1/s}\}.\end{equation}
\end{proof}

We remark that for the case when $w\equiv 1$ and $p=\infty,$ we can take $\kappa=\infty$, $\eta=1,$ and so from
\eqref{C1def} we deduce the separation estimate
$$
\delta(\omega_N^*)\ge C_2N^{-1/\alpha}\qquad (N\ge 2),
$$
where
\begin{equation}\label{C2def}
C_2:=\left[\frac{\bar{\mu}(A)}{C_0(r_1)} (1-\alpha/s)\right]^{1/\alpha}(\alpha/s)^{1/s},\  r_1=\diam(A).
\end{equation}

For the proof of  Theorem~\ref{let}, we utilize the following.
\begin{lemma}\label{lower-control} Let $A$ be a compact, infinite,  lower $\alpha$-regular  metric space
with lower  $\alpha$-regular measure $\underline{\mu}$, $w: A\times A\rightarrow[0,\infty)$ be an SLP weight on $A$, and $s> \alpha$.   Then there exists a positive integer $N_0$ independent of $s$,  such that
\begin{equation}\label{lemma-1}\E_{s}^{w}(N,A)\geq  C_5  N^{1+s/\alpha}\qquad (N\ge N_0),
\end{equation}
where $C_5$ is a constant
independent of $N$ given below in \eqref{C5def}.

\end{lemma}

\begin{proof}

Let $\kappa$ and $\eta$ be as in \eqref{etakappa} and let $0<r_2\le \kappa$.  Since $A$ is compact, there is some $M$ such that the $M$-point best-packing distance satisfies
\begin{equation}\label{delta-on-N}\delta_{M}(A)\le r_2.\end{equation}

Let $N>M$ and let $\omega_N=\{x_1,\ldots, x_N\}\subset  A$ be an arbitrary $N$-point configuration of distinct points.  For $1\le i\le N$, let $y_i\in\omega_N$ be a fixed nearest neighbor to $x_i$ in the configuration $\omega_N$, and  set
$$\delta_i:=\d(x_i,y_i)=\min_{\substack{1\leq j\leq N\\  j\not=i}}\d(x_i,x_j)>0.$$
We assume an ordering on $\omega_N$ so that $\delta_i\le \delta_{i+1}$ for $i=1,\ldots,N-1$.
We note that $\omega_N\setminus\{x_1,\ldots,x_{N-M}\}$ is of cardinality $M$ and thus for all $i\leq N':= N-M$ we have that $\delta_i\leq r_2\leq \kappa$.

The energy of $\omega_N$ then has the lower bound
\begin{align}\label{inequality-1}
\begin{split}
E_{s}^{w}(\omega_N)&\geq\sum_{i=1}^{N'}\frac{w(x_i,y_i)}{\delta_i^s}\ge\sum_{i=1}^{N'}\eta\left(\frac{1}{\delta_i^\alpha}\right)^{s/\alpha}
\ge  \eta\left(\sum_{i=1}^{N'}\frac{1}{\delta_i^\alpha}\right)^{s/\alpha}(N')^{1-s/\alpha}
\\ &\ge
\eta\left(\sum_{i=1}^{N'} \delta_i^\alpha\right)^{-s/\alpha}(N')^{1+s/\alpha}
=\eta 2^{-s}\left(\sum_{i=1}^{N'}\left(\frac{\delta_i}{2}\right)^\alpha\right)^{-s/\alpha}(N')^{1+s/\alpha}.\end{split}
\end{align}where the last inequality in the first line follows from  Jensen's inequality  and  the subsequent inequality   follows from the harmonic-arithmetic mean  inequality.

Let  $\Lambda>1$ and $N_0:=M\Lambda/(\Lambda-1)$.  Then $N'= N-M\ge \Lambda^{-1}N$ for $N\ge N_0$.
Noting that the balls $B(x_i,\delta_i/2)$ are pairwise disjoint, we may apply the lower regularity of $\underline{\mu}$ (with regularity constant $c_0(r_2)$) to  obtain
\begin{align}\label{last}
\begin{split}E_{s}^{w}(\omega_N)&\geq \eta 2^{-s}\left(c_0(r_2)\sum_{i=1}^{N'}\underline{\mu}\left(B(x_i,\frac{\delta_i}{2})\right)\right)^{-s/\alpha}(N')^{1+s/\alpha}\\
&\ge\frac{ \eta}{(2^\alpha\, c_0(r_2)\, \underline{\mu}(A))^{s/\alpha}}(N')^{1+s/\alpha}\\&
\ge\Lambda^{-1-s/\alpha}\frac{ \eta}{(2^\alpha\, c_0(r_2)\, \underline{\mu}(A))^{s/\alpha}}N^{1+s/\alpha}\end{split}
.\end{align}
Since \eqref{last} holds for arbitrary $N$-point configurations $\omega_N\subset  A$ with $N\ge N_0$,  we obtain that  \eqref{lemma-1} holds with
\begin{equation}
\label{C5def}
C_5:=\Lambda^{-1-s/\alpha}\, \eta\, 2^{-s}\,{( c_0(r_2)\,  \underline{\mu}(A))^{-s/\alpha}}.
\end{equation}
We remark
that $N_0$ depends on $\Lambda$ and $r_2$, but is independent of $s$.
 \end{proof}

\begin{proof}[Proof of Theorem~\ref{let}]  Appealing to the generality provided by Theorem~\ref{separation} and Lemma~\ref{lower-control},  we can substantially extend and improve upon the arguments used in the proof of Theorem 3.6 in \cite{maym}.

Let $\omega_N^*=\{x_1\ldots,x_N\}$ be an $N$-point $(s,w)$-energy minimizing  configuration for  the compact set $K$, and, for $y\in K$,   consider the  function
\begin{equation}\label{point}U(y):=\frac{1}{N}\sum_{i=1}^N \frac{w(y,x_i)}{\d(y,x_i)^s}.\end{equation}
For fixed $1\leq j\leq N$, the function $U(y)$ can be decomposed as
\begin{equation}\label{decomp}U(y)=\frac{1}{N}\frac{w(y,x_j)}{\d(y,x_j)^s}+\frac{1}{N}\sum_{\substack{i=1\\i\not=j}}^N\frac{w(y,x_i)}{\d(y,x_i)^ s},\end{equation}
and, since $\omega_N^*$ is a minimizing configuration on $K$, the point $x_j$  minimizes the sum over $i\not=j$ on the right-hand side of equation~\eqref{decomp}. Thus for each fixed $j$ and $y\in K$
\begin{align}\label{lower}
U(y)&\geq \frac{1}{N}\frac{w(y,x_j)}{\d(y,x_j)^ s}+\frac{1}{N}\sum_{\substack{i=1\\i\not=j}}^N\frac{w(x_j,x_i)}{m(x_j,x_i)^ s}.
\end{align}
Summing over $j$ gives
 \begin{align}
 N U(y)&\geq\frac{1}{N}\sum_{j=1}^N\frac{w(y,x_j)}{\d(y,x_j)^ s}+\frac{1}{N}\sum_{j=1}^N \sum_{\substack{i=1\\i\not=j}}^N\frac{w(x_j,x_i)}{m(x_j,x_i)^ s}
\\&= U(y)+\frac{1}{N}\E_{s}^{w}(N,K),
\end{align}
and thus
\begin{equation}
\label{Ulowbnd} U(y)\geq\frac{1}{N(N-1)}\E_{s}^{w}(N,K)\ge \frac{\E_{s}^{w}(N,K)}{N^2} \qquad (y\in K).\end{equation}

Since $K$ is compact, there exists a point $y^*\in K$ such that
\begin{equation}\label{y}\min_{1\leq i\leq N}\d(y^*,x_i)=\rho(\omega_N^*,K)=:\rho(\omega_N^*).
\end{equation}

Using the fact that a  function is lower semi-continuous if and only if it is the limit of an increasing sequence of continuous functions, it is not difficult to show that since $w$ is a bounded SLP weight on  $K$, it may be extended to a bounded SLP weight on $A$.
Then, by  Lemma~\ref{lower-control}, there are constants $N_0$ and $C_5>0$ such that
 \begin{equation}\label{EKlowbnd}\E_{s}^{w}(N,K)\ge \E_{s}^{w}(N,A)\geq  C_5  N^{1+s/\alpha}\qquad (N\ge N_0).
\end{equation}
We note that the constant $C_5$ of~\eqref{EKlowbnd} does not depend on $K$, but rather on $A$ (specifically on the lower regularity constant of $A$ and on $\mu(A)$) as well as on the extended weight $w$.

Since~\eqref{Ulowbnd} holds for the point $y^*$ of~\eqref{y}, we combine \eqref{Ulowbnd} with~\eqref{EKlowbnd} to obtain
\begin{equation}\label{lower-pot}U(y^*)\geq  \frac{\E_{s}^{w}(N,K)}{N^2} \ge C_5N^{s/\alpha-1} \qquad (N\ge N_0).
\end{equation}

Next we determine an upper bound for $U(y^*)$ using the  $\alpha$-regularity of the superset $A$.
Since $A$ is upper $\alpha$-regular, we see that $K$ is also because ${\mu}(K)>0$.  Hence,
 Corollary~\ref{seps} applied  to $K$ implies that  there is some $C_2>0$ such that $ \delta(\omega^*_N)\ge C_2N^{-1/\alpha}$ for   $N\ge 2$.
 We note that the constant $C_2$ here depends on  $K$, specifically  $\mu(K)$.

Let $\mathcal{N}$ consist of those $N\ge N_0$ such that
\begin{equation}\label{reverse}\rho(\omega_N^*)\geq \frac{C_2}{2}N^{-1/\alpha}.
\end{equation}
If $\mathcal{N}$ is empty (or finite) then we are done.   Assuming that $\mathcal{N}$ is nonempty,
let $N\in \mathcal{N}$ be fixed.

 For $0<\epsilon<1/2$, let
\begin{equation}\label{r0def}
r_0=r_0(N,\epsilon):=\epsilon\,  {C_2}N^{-1/\alpha}.
\end{equation} Note that any two of the balls $B(x_i,r_0)\subset  A$, for $1\leq i\leq N,$ do not intersect since $r_0<{\delta(\omega_N^*)}/{2}$.

 For any $x\in B(x_i,r_0)$, inequalities~\eqref{y} and \eqref{reverse} imply
\begin{align}\label{gooba}
\begin{split}
\d(x,y^*)&\leq \d(x,x_i)+\d(x_i,y^*)\leq r_0+\d(x_i,y^*)\\ &\leq  2\epsilon\, \rho(\omega_N^*) +\d(x_i,y^*)\leq (1+ 2\epsilon)\d(x_i,y^*).\end{split}
\end{align}

For fixed $1\leq i\leq N$, using ~\eqref{gooba} and  taking an average value on $B(x_i,r_0)$ we obtain
\begin{align}\label{mort}
\begin{split}\frac{w(x_i,y^*)}{\d(x_i,y^*)^s}&\leq \frac{ \|w\|_\infty(1+ 2\epsilon)^s}{ \mu(B(x_i,r_0))}
\int_{B(x_i,r_0)}\frac{d\mu(x)}{\d(x,y^*)^s}\\ &\leq  \frac{ \|w\|_\infty\,(1+ 2\epsilon)^s\,  c_0(r_0)}{ \,r_0^{\alpha}}
\int_{B(x_i,r_0)}\frac{d\mu(x)}{\d(x,y^*)^s},\end{split}
\end{align}
where $\|w\|_\infty$ denotes the sup-norm of $w$ on $A\times A$  and $c_0(r_0)$ is the localized constant of \eqref{localreg}  for the set $A$.

Inequality  \eqref{reverse} and definition \eqref{r0def} imply $2\epsilon \rho(\omega_N^*)\ge r_0$  and thus,
for $x\in B(x_i,r_0)$, we obtain
\begin{align}\label{gabba}
\begin{split}
\d(x,y^*) &\geq\d(x_i,y^*)-\d(x,x_i) \geq \d(x_i,y^*)-r_0\\ &\geq  \d(x_i,y^*)-  2\epsilon \, \rho(\omega_N^*)\geq (1- 2\epsilon)\rho(\omega_N^*) . \end{split}
\end{align}
Inequality~\eqref{gabba} implies $$\bigcup_{i=1}^NB(x_i, r_0)\subset  A\setminus B(y^*,(1- 2\epsilon)\rho(\omega_N^*)),$$and since the left-hand side is a disjoint union, averaging the inequalities of~\eqref{mort}   we
have
\begin{align}\label{nadamas}
\begin{split}U(y^*)&\leq \frac{ \|w\|_\infty\,(1+ 2\epsilon)^s\,  c_0(r_0)}{ N\,r_0^{\alpha}}
\sum_{i=1}^N\int_{B(x_i,r_0)}\frac{d\mu(x)}{\d(x,y^*)^s}\\
&\leq\frac{ \|w\|_\infty\,(1+ 2\epsilon)^s\,  c_0(r_0)}{ N\,r_0^{\alpha}}\int_{A\setminus B(y^*, (1- 2\epsilon)\rho(\omega_N^*))}\frac{d\mu(x)}{\d(x,y^*)^s}. \end{split}\end{align}
%


For fixed $\tau\geq1$ we define the radius $R(N):= \tau(1- 2\epsilon)\rho(\omega_N^*),$ and the constant
 \begin{equation}\label{empty-ness}\tilde{C}_0(\tau):=
 C_0(R(N))(1-\tau^{\alpha-s})+ C_0\tau^{\alpha-s}.\end{equation}Note that if $\tau=1$, then $\tilde{C}_0(1)=C_0$.  (We retain $\tau$ as a parameter in our estimates as an option for the reader to optimize $C_3$ for a fixed $s$.)
Now we break the integral on the right-hand side of~\eqref{nadamas} into two terms and proceed as in~\eqref{truncate} to obtain
\begin{align}\label{needed}\begin{split}
 \int_{A\setminus B(y^*, (1- 2\epsilon)\rho(\omega_N^*))}&\frac{d\mu(x)}{\d(x,y^*)^s}
 \\ &=  \int_{  B(y^*, (1- 2\epsilon)\rho(\omega_N^*),R(N))}\frac{d\mu(x)}{\d(x,y^*)^s}+\int_{A\setminus B(y^*, R(N))}\frac{d\mu(x)}{\d(x,y^*)^s}\\
 & \leq C_0(R(N))\int_{R(N)^{-s}}^{[(1-2\epsilon)\rho(\omega_N^*)]^{-s}}t^{-\alpha/s}dt
 +C_0\int_0^{R(N)^{-s}}t^{-\alpha/s}dt\\
& =\frac{\tilde{C}_0(\tau)}{(1-\alpha/s) (1-2\epsilon)^{s-\alpha}} \rho(\omega_N^*)^{\alpha-s}.
\end{split}\end{align}

It is convenient to define the quantity
\begin{equation}\label{epthetacoef}\beta(\epsilon):=\frac{\|w\|_\infty(1+2\epsilon)^s }
{(1-\alpha/s)(1-2\epsilon)^{s-\alpha}(\epsilon C_2)^\alpha},
\end{equation}
and we note that  for fixed $s>\alpha$ it is minimized as a function of $\epsilon$ for
\begin{equation}\label{correct-epsilon}\epsilon_0:=\frac{1}{2(2(s/\alpha)-1)}<\frac{1}{2},\end{equation}
with minimal value
\begin{equation}\label{minimal-beta}\beta_0:=\beta(\epsilon_0)= \frac{\|w\|_\infty}{(1-\alpha/s)^{s-\alpha+1}}\left(\frac{4s}{\alpha C_2}\right)^{\alpha}.\end{equation}

Using $\epsilon_0$ and combining inequality~\eqref{nadamas} with inequality~\eqref{needed}  we obtain
\begin{equation}\label{int2est}U(y^*)  \leq
c_0(r_0)\beta_0\tilde{C}_0(\tau) \rho(\omega_N^*)^{\alpha-s}.
\end{equation}
If $N\in \mathcal{N}$, then    \eqref{int2est} and \eqref{lower-pot}  imply
$$\rho(\omega_N^*)\le \left[\frac{c_0(r_0)\beta_0 \tilde{C}_0(\tau)}{ C_5}
  \right]^{1/(s-\alpha)} N^{-1/\alpha}. $$
If $N\not\in\mathcal{N}$, then either $N\le N_0$ or $\rho(\omega_N^*)<\frac{C_2}{2} N^{-1/\alpha}$.
Hence \eqref{mesh-1} holds with \begin{equation}\label{C3def}
C_3:= \max\left\{ \diam(A)N_0^{1/\alpha},\,\left[\frac{c_0(r_0)\beta_0 \tilde{C}_0(\tau)}{ C_5}
  \right]^{1/(s-\alpha)}, \frac{C_2}{2}\right\}.
\end{equation}
We note that if $N>N_0$, then it suffices to take
\begin{equation}\label{C3N0}C_3= \max\left\{\left[\frac{c_0(r_0)\beta_0 \tilde{C}_0(\tau)}{ C_5}
  \right]^{1/(s-\alpha)}, \frac{C_2}{2}\right\}
  \end{equation}

\end{proof}

\begin{proof}[Proof of Theorem~\ref{QUBP}]
Starting with Theorem~\ref{let}  we shall employ a bootstrapping argument whereby the constants $C_2,$ $C_5$, and subsequently $C_3$
are redefined so as to depend
on $N$.

We begin by noting that if $s\geq 2\alpha$, then the constant $C_3$ of~\eqref{C3def} has a uniform upper bound in $s$; indeed, with $\kappa=\infty, C_2$ as defined in \eqref{C2def} and $C_5$ as defined in \eqref{C5def} (with $\eta=1$), each of the
three terms appearing in braces in \eqref{C3def} is uniformly bounded above. Thus
there exists a constant $C^*$ independent of $N\geq2$ and of $s\geq2\alpha$ such that
$\rho(\omega_N^{(s)},K)< C^* N^{-1/\alpha},$  where $\omega_N^{(s)}$ is any $N$-point $(s,1)$-energy minimizing configuration on $K.$\

We next note that $C_0(0)$ of~\eqref{localreg-0} is finite and positive, and  utilizing the constant $c_A$ of~\eqref{seplowbnd}  we fix
\begin{equation}\label{good-constant}C^{**}:=\max\left\{C^*,\ c_A,\ \left(\frac{\mu(K)}{C_0(0)}\right)^{1/\alpha}\right\},
\end{equation}
and we now redefine the radius $r_1$ to be a function of $N$, \begin{equation} r_1(N):=C^{**}\,N^{-1/\alpha}\quad (N\geq2).\end{equation}
Returning to  the proof of Theorem~\ref{separation}, we note that $r_1(N)>\rho(\omega^{(s)}_N,K)$, and so inequality~\eqref{r1def} holds.  Furthermore, by the choice of $C^{**}$ we have that for  $0<\theta_0<1$ as in~\eqref{theta-0}
$$r_0(N):=\left(\frac{\theta_0\mu(K)}{N C_0(0)}\right)^{1/\alpha}<r_1(N).$$Taking $r_0=r_0(N)$ in the proof  and remembering that $q=1$ in the current context, we see that with $A$ replaced by $K$    the penultimate term on  right-hand side of~\eqref{truncate} becomes
$$\frac{sC_0(r_1(N))}{s-\alpha}\left(\frac{\theta_0\,\mu(K)}{N\,C_0(0)}\right)^{1-s/\alpha},$$
and thus
\begin{align}\begin{split}
 \int_{ B(x_j,r_0(N),r_1(N))}\frac{d\mu(x)}{\d(x,x_j)^s}&\leq \frac{sC_0(r_1(N))}{s-\alpha}\left(\frac{\theta_0\,\mu(K)}{N\,C_0(0)}\right)^{1-s/\alpha} \\  &\leq \frac{s}{s-\alpha}\left(\frac{\theta_0\,\mu(K)}{N}\right)^{1-s/\alpha}\,C_0(r_1(N))^{s/\alpha}, \end{split}\end{align}
where the last inequality follows from the fact that $C_0(0)\leq C_0(r_1(N))$ and $s>\alpha$.

For $w\equiv1$, the constant $C_2$ of \eqref{C2def} with $r_1=r_1(N)$ becomes
\begin{equation}\label{need-later}C_2(N):=\left(\frac{\alpha}{s}\right)^{1/s}\left(\frac{ 1-\alpha/s}{C_0(r_1(N))}\right)^{1/\alpha}\mu(K)^{1/\alpha},\end{equation}where $C_0(r_1(N))$ is the local upper regularity constant of~\eqref{localreg},
and we have
$$
\delta(\omega_N^{(s)})\ge C_2(N)N^{-1/\alpha}\qquad (N\ge 2, \, s\ge 2\alpha).
$$

Furthermore, allowing the radius $r_2$ appearing in~\eqref{delta-on-N} to depend on $N\geq 2$ by taking $r_2:=r_1(N),$   we see via~\eqref{seplowbnd} and~\eqref{good-constant} that $$r_1(N)\geq\delta_N(A) \quad (N\geq2),$$and there is no need to designate the integer $M$ in the proof of Lemma~\ref{lower-control}. Thus we can take $\Lambda=1$ in~\eqref{C5def}, and it follows (with $\eta=1$) that
$$ E^1_s(\omega_N^{(s)}) \geq C_5(N)N^{1+s/\alpha} \qquad (N\ge 2,\,\, s \geq 2\alpha),
$$
where

\begin{equation}\label{C5_N}C_5(N):=\frac{1}{2^s[c_0(r_1(N))\mu(A)]^{s/\alpha}}.\end{equation}

We remark that $C_2(N)$ clearly depends on the subset $K$, whereas $C_5(N)$ depends on the superset $A$.

We now return to the proof of Theorem~\ref{let} utilizing the constants $C_2(N)$ and $C_5(N)$.
For   $\beta_0$ as in~\eqref{minimal-beta}, we see that
$$ \rho(\omega^{(s)}_N,K) \leq C_3(N)N^{-1/\alpha} \,\,\,(N\geq N_0,\,\, s\geq2\alpha),$$
where $N_0$ is as in Lemma~\ref{lower-control}, and by~\eqref{C3N0} (choosing $\tau=1$, so that $\tilde{C}_0(\tau)=C_0$)
\begin{equation}\label{C3-Def}C_3(N):=\max\left\{ \left[\frac{c_0(r_0)\beta_0C_0}{ C_5(N)}
  \right]^{1/(s-\alpha)}, \frac{C_2(N)}{2} \right\}.\end{equation}


With equations~\eqref{need-later}-\eqref{C3-Def} in mind, we are ready to complete the proof of Theorem~\ref{QUBP}.  The   argument leading to equation~\eqref{nuN} shows that $\nu_N$ is an $N$-point best-packing configuration on $K$ for each $N\ge 2$.
We now need to determine the limits of the constants $C_2(N)$ of~\eqref{need-later} and $C_3(N)$ of~\eqref{C3-Def} as $s\to \infty$. Fixing $N$ in~\eqref{need-later}  yields
\begin{equation}\label{c2s}
\lim_{s\to\infty} C_2(N)=\left(\frac{\mu(K)}{C_0(r_1(N))}\right)^{1/\alpha}=:\hat{C}_2(N).
\end{equation}
Since $c_0(r_0)$ and $C_0$ are independent of $s$ and $\lim_{s\to\infty}\beta_0^{1/(s-\alpha)}=1$, it  follows, that for fixed $N$
\begin{align}\label{too-live-crew}
\begin{split}
\lim_{s\to\infty} C_3(N)&=  \max\left\{ \frac{\hat{C}_2(N)}{2},\, \lim_{s\to\infty}C_5(N)^{1/(\alpha-s)}\right\}\\ &=\max\left\{\frac{1}{2}\left(\frac{\mu(K)}{C_0(r_1(N))}\right)^{1/\alpha},\, 2[c_0(r_1(N))\mu(A)]^{1/\alpha}\right\}\\&:=\hat{C}_3(N)\end{split}.
\end{align}

From the continuity of $\delta(\cdot)$ and $\rho(\cdot,K)$ on $K^N$ we deduce that
$$\delta(\nu_N)\ge \hat{C}_2(N) N^{-1/\alpha} \quad \text{and} \quad \rho(\nu_N,K)\le \hat{C}_3(N)N^{-1/\alpha} \qquad (N\ge N_0).$$
Taking the ratio of these two quantities we have that
\begin{equation}\label{doggone}\frac {\rho(\nu_N,K)}{\delta(\nu_N)}\leq\frac{\hat{C}_3(N)}{\hat{C}_2(N)}=\max\left\{\frac{1}{2},\,2\left(\frac{\mu(A)}{\mu(K)}\right)^{1/\alpha}[c_0(r_1(N))\,C_0(r_1(N))]^{1/\alpha}\right\}, \end{equation}
and hence for $N\geq N_0$
\begin{align}
\limsup_{N\to\infty}\frac {\rho(\nu_N,K)}{\delta(\nu_N)}&\leq \max\left\{\frac{1}{2},\,2\left(\frac{\mu(A)}{\mu(K)}\right)^{1/\alpha}[c_0(0)\,C_0(0)]^{1/\alpha}\right\}\\
&=2\left(\frac{\mu(A)}{\mu(K)}\right)^{1/\alpha}[c_0(0)\,C_0(0)]^{1/\alpha}<\infty.
\end{align}
Therefore,
the sequence of configurations $\{\nu_N\}_{N=2}^\infty$ is quasi-uniform on $K$.
\end{proof}

\noindent
{\bf Acknowledgments.} We thank the referees for their careful reading and detailed comments.

\end{document}